\documentclass[twoside,leqno,twocolumn,11pt]{article}
\usepackage{ltexpprt}
\usepackage{microtype}
\usepackage{graphicx}
\usepackage[algo2e,ruled]{algorithm2e}
\usepackage{color}
\usepackage{numprint}
\usepackage{subfig}
\usepackage[left]{lineno}
\usepackage{amsmath}
\usepackage{amssymb}
\usepackage{mathtools}
\usepackage{numprint}
\usepackage{widetable}
\usepackage{tabularx}
\usepackage{url}
\usepackage{pgfplots}
\usepackage{xcolor}
\usepackage{verbatim}
\usepackage{footnote}
\usepackage{wrapfig}
\usepackage{placeins}
\usepackage{float}

\definecolor{bblue}{HTML}{4F81BD}
\definecolor{rred}{HTML}{C0504D}
\definecolor{ggreen}{HTML}{9BBB59}
\definecolor{ppurple}{HTML}{9F4C7C}

\nprounddigits{2}

\newcommand{\bigO}[1]{\mathcal{O} \bigl (#1 \bigr )}
\newcommand{\norm}[1]{\left\lVert#1\right\rVert}

\begin{document}

\title{Estimating Current-Flow Closeness Centrality \\ with a Multigrid Laplacian Solver}
\author{Elisabetta Bergamini\thanks{Institute of Theoretical Informatics, Karlsruhe \newline \hspace*{1.5em} Institute of Technology (KIT), Karlsruhe, Germany} \and Michael Wegner\footnotemark[1] \and Dimitar Lukarski\thanks{Paralution Labs UG \& Co. KG, Gaggenau, Germany} \and Henning Meyerhenke\footnotemark[1]}

\date{}

\maketitle

\begin{abstract}
Matrices associated with graphs, such as the Laplacian, lead to numerous interesting graph problems expressed as linear systems. One field where Laplacian linear systems play a role is network analysis, e.\,g.\ for certain centrality measures that
indicate if a node (or an edge) is important in the network. One such centrality measure is current-flow closeness.

To allow network analysis workflows to profit from a fast Laplacian solver, we provide an implementation of
the LAMG multigrid solver in the NetworKit package, facilitating the computation of current-flow closeness 
values or related quantities.
%
Our main contribution consists of two algorithms that accelerate the current-flow computation for
one node or a reasonably small node subset significantly. One sampling-based algorithm provides an
unbiased estimation of the related electrical farness, 
the other one is based on the Johnson-Lindenstrauss transform.
Our inexact algorithms lead to very accurate results in practice. Thanks to them one is now able to compute an estimation of current-flow closeness of one node on networks with tens of millions of nodes and edges within seconds or a few minutes.
From a network analytical point of view, our experiments indicate that current-flow closeness can discriminate among different nodes significantly better than traditional shortest-path closeness and is also considerably more resistant to noise -- we thus show
that two known drawbacks of shortest-path closeness are alleviated by the current-flow variant.\\[0.75ex]
\end{abstract}

\section{Introduction}

Laplacian linear systems $Lp = b$ play a central role in various tasks in algorithmic graph theory and network analysis.
Particularly the connection between the graph Laplacian $L$ and electrical networks allows analytical and algorithmic
insights~\cite{mavroforakis2015spanning,DBLP:conf/stacs/BrandesF05}. Based on this connection, several centrality measures, which aim at identifying the most important nodes or edges in a network, have been introduced in the literature. For example, spanning edge centrality~\cite{mavroforakis2015spanning} measures the  importance of an edge for the connectedness of the graph. Current-flow betweenness~\cite{DBLP:conf/stacs/BrandesF05} indicates the participation of a node in paths between other nodes, while current-flow closeness~\cite{DBLP:conf/stacs/BrandesF05} measures the average distance from a node to the other nodes of the network -- both in terms of an electrical distance measure. Differently from shortest-path closeness, here distance is a quantity which takes \emph{all} paths between two nodes into account. In addition to centrality measures, Laplacian linear systems can be used for several other tasks in algorithmic graph theory and algorithmic network analysis, including sparsification~\cite{DBLP:journals/siamcomp/SpielmanS11}, graph partitioning~\cite{DBLP:journals/pc/MeyerhenkeMS09}, 
approximate maximum network flow~\cite{DBLP:conf/stoc/ChristianoKMST11} and graph drawing~\cite{DBLP:journals/tvcg/GansnerHN13}.

Although algorithms for solving Laplacian linear systems quickly in practice have been proposed, such as the multigrid-based solvers LAMG~\cite{livne2012lean} and CMG~\cite{koutis2011combinatorial}, an implementation of these solvers is not available in popular network-analysis frameworks, such as NetworkX~\cite{Hagberg2008}, igraph~\cite{igraph}, graph-tool~\cite{peixoto_graph-tool_2014} and NetworKit~\cite{Staudt2014}. In this paper, we implement LAMG in NetworKit, bridging the gap between performance-focused frameworks for network analysis and state-of-the-art Laplacian solvers. Also, we formulate the current-flow closeness centrality presented in~\cite{DBLP:conf/stacs/BrandesF05} in terms of linear systems between pairs of nodes and propose two approximation\footnote{Approximation in the sense of being inexact, not necessarily with guarantee on the approximation quality.} algorithms which work very well in scenarios where the centrality of a single node or a subset of nodes has to be computed, according to our experimental study. 
To the best of our knowledge, only the NetworkX~\cite{Hagberg2008} package contains an implementation of current-flow closeness. The approach implemented in NetworkX requires to invert the Laplacian of the graph, which takes $\Omega(n^2)$ time (and cubic time in practice with standard tools), since the inverse of the Laplacian is in general a dense matrix. This does not allow to scale to large networks with millions of nodes and edges, even when we want to compute the closeness of a single node or of a subset of nodes. With this approach, in fact, computing the closeness of one node is just as expensive as computing it for all nodes. On the contrary, our approach can estimate the closeness of a single node very quickly: For example, it requires less than 2 minutes on a network with 50 millions edges.

Thanks to our approach, we can study for the first time the properties of current-flow closeness in large networks with tens of millions of nodes. We compare current-flow closeness with traditional shortest-path closeness and show that the former succeeds in differentiating nodes significantly better than the latter, and is also more resilient to noise. In addition, we study the correlation between centrality measures and degrees in real-world networks, in relation to a recent theoretical result for random geometric graphs~\cite{DBLP:journals/jmlr/LuxburgRH14}. Our experiments show that there is a strong correlation between degrees and current-flow closeness in complex networks, whereas there is basically no correlation in street networks.

This update to previous versions (SIAM CSC 2016, arXiv) fixes the statement of Proposition~\ref{prop:unbiased}
as well as a few minor issues. Content-wise it still reflects the state of late 2016.

\section{Preliminaries}
\label{sec:prelim}
\paragraph{Graph basics.} Throughout the paper we consider connected undirected graphs $G = (V, E, w)$ having $n = |V|$ nodes and $m = |E|$ edges. The function $w : E \rightarrow \mathbb{R}$ assigns a weight $w(e)$ to each edge $e = \lbrace u, v \rbrace$ where $u$ and $v$ are the nodes connected by the edge. 
We introduce the shorter notation $\omega_{uv}$ for the weight $\omega(e)$ of an edge $e = \{u,v\}$.
Disconnected graphs require some adjustments of the proposed algorithmic approach, but can in principle 
be handled by treating each connected component separately.
We denote by $\vec{E}$ the set of bidirected edges of $G$, i.e. $\vec{E} := \{(u,v) : u, v \in V, \{u,v\} \in E \}$.

We say $\pi = (v_1, ..., v_k)$ is a \textit{path} in $G$ if $v_i \in V$ for $i =1,...,k$ and $\{v_i, v_{i+1}\} \in E$ for $i =1, ..., k-1$. We denote the quantity $\sum_{i=1}^{k-1} \omega_{v_i, v_{i+1}}$ as the \textit{length} of path $\pi$. The path(s) of minimum length among all paths between two nodes $u$ and $v$ is (are) called the \textit{shortest path(s)} between $u$ and $v$.

A Laplacian matrix $L$ is defined for an undirected graph $G$ as $L := D-A$, where $A=A(G)$ is the (weighted) adjacency matrix
of $G$ and $D=D(G)$ the diagonal matrix storing the (weighted) node degrees: $D_{ii} = \sum_{j = 1}^{n} \omega_{ij}$.

\paragraph{Graphs as electrical networks.}
\label{par:elec-nets}
One can regard a graph as an \textit{electrical network} where each edge $\{u,v\}$ corresponds to a resistor
with conductance $\omega_{uv}$ (the edge weight) or resistance $1 / \omega_{uv}$.  
We can interpret the conductance as the ease with which an electrical current can flow through the edge. We can associate a \textit{supply} $b : V \rightarrow \mathbb{R}$ with the electrical network, representing the nodes where current enters and leaves the network. A positive supply $b(v)$ means that current is entering the network from node $v$ and a negative supply means that current is leaving the network. In the following, we will always assume that $\sum_{v \in V} b(v) = 0$ and that $b(s) = +1$ and $b(t) = -1$ for two nodes $s$ and $t$, and that $b(w) = 0\ \forall w \neq s,t$. 
Also, we will refer to such a supply as vector $b_{st} \in \mathbb{R}^{n \times 1}$.
We could interpret this as $s$ and $t$ being the two poles of a battery: this generates a current $e_{st}: {\vec{E}} \rightarrow \mathbb{R}$ flowing through the network (if seen from the other direction of an edge, the current changes its sign). To each node $v$ we can associate a \textit{potential} $p_{st}(v)$ such that the vector $p_{st} \in \mathbb{R}^{n \times 1}$ satisfies the following linear system:
 \vspace{-1ex}
\begin{equation}
\label{eq:system}
Lp_{st} = b_{st}
\end{equation}
Then, the current flowing through edge $(u,v)$ is defined as $(p_{st}(u) - p_{st}(v))/ \omega_{st}$. Notice that, since $G$ is connected, the rank of the Laplacian is $n-1$ and there are infinitely many vectors $p_{st}$ satisfying Eq. (\ref{eq:system}), each of them differing from the other by an additive constant. However, the current is well defined, since it depends on the difference between two potentials.

\section{Related Work}
\label{sec:rel_work}

\subsection{Solving Laplacian linear systems.}
We focus our description on iterative solvers due to their better time complexity on sparse graphs compared to direct solvers.
Most iterative solvers reduce the norm of the residual $r = \norm{b - Ax}$ iteratively by altering the current preliminary solution vector $x$ in every iteration. One usually stops when the (relative) residual is below a certain
tolerance $\tau$, which yields a vector $x'$ that is a good enough approximation to the actual solution $x$. 
%
While there are recent advances in theory to solve special linear systems including Laplacians in nearly-linear time~\cite{Spielman2004, Kelner2013}, those algorithms are not competitive in practice yet~\cite{hoske2015, boman2015evaluating}. 
In fact, the Conjugate Gradient (CG) algorithm outperforms the nearly-linear time algorithms in practice even though its asymptotic running time is typically higher.

A popular class of iterative algorithms to solve linear equations quickly in practice is called Algebraic Multigrid (AMG)~\cite{ruge1987algebraic}. The basic idea is to solve the actual linear system by iteratively solving coarser (i.e. smaller) yet similar systems and projecting the solutions of those back to the original system. AMG algorithms can be distinguished by the class of matrices they can handle and the way they construct the coarser systems. Two fast algorithms that are specifically designed for solving Laplacian systems are CMG by Koutis et al.~\cite{koutis2011combinatorial} and LAMG by Livne and Brandt~\cite{livne2012lean}.

 
We decided to use LAMG as linear solver due to its particular design for complex networks. To this end, LAMG alternates between two stages called \emph{elimination} and \emph{aggregation} to construct the coarser systems. The former eliminates low degree nodes in the corresponding graph, the latter partitions nodes into \emph{aggregates} based on a special affinity measure~\cite{livne2012lean}. Both stages reduce the number of nodes and thus define a coarsening mechanism.
Based on an extensive evaluation, Livne and Brandt state that running times of LAMG and CMG are comparable but LAMG tends to be more robust in the sense that CMG has large outliers on a small set of systems~\cite{livne2012lean}.

\subsection{Laplacian linear systems for network analysis.} \label{subsec:LaplaciansInNetworkAnalysis}
The connection between the graph Laplacian and electrical networks (see Section~\ref{par:elec-nets}) has allowed for the solution of several graph algorithmic problems in terms of Laplacian systems. One of them is a centrality measure called \emph{spanning edge centrality}, which indicates whether an edge is vital for the
connectedness of a network~\cite{mavroforakis2015spanning}. The notion of importance for connectedness is also helpful for graph sparsification. A sparsification algorithm takes a dense graph and wants to find a sparser representation (= with fewer edges)
with the same vertex set and similar properties~\cite{DBLP:journals/siamcomp/SpielmanS11}, e.\ g.\ approximately the same cut sizes
or eigenvalues.
Since processes described by Laplacian linear systems can distinguish sparse from dense graph regions,
edges in dense areas are, intuitively speaking, redundant and can be ``sparsified'' without doing much 
harm to the cut sizes when the weights of retained edges are properly scaled.
This idiosyncrasy allows the use of processes described by Laplacian linear systems also for graph partitioning~\cite{DBLP:journals/pc/MeyerhenkeMS09}, 
approximate maximum network flow~\cite{DBLP:conf/stoc/ChristianoKMST11}, and graph drawing~\cite{DBLP:journals/tvcg/GansnerHN13}. 
Moreover, the connection to electrical flow makes the use in dynamic load balancing of divisible tokens by diffusion~\cite{DiekmannFrommerMonien99efficient} possible.

The interpretation of a graph as an electrical network has also led to the definition of two centrality measures based on current flow, current-flow closeness and current-flow betweenness~\cite{DBLP:conf/stacs/BrandesF05}. Compared to traditional closeness and betweenness centrality, these two measures take \emph{all} paths between two nodes into account and not only \emph{shortest} paths.

\section{Current-flow closeness centrality}
\label{sec:cf_closeness}
Closeness centrality measures the efficiency of a node in spreading information to the other nodes of the network. Formally, let $d_{\text{SP}}(u,v)$ the shortest-path distance between $u$ and $v$ (i.e.\ the length of the shortest path(s) between $u$ and $v$). Closeness of node $v$ is then the reciprocal of the average shortest-path distance between $v$ and the other nodes:
\begin{equation}
\label{eq:closeness}
c_{\text{SP}}(v) := \frac{n-1}{\sum_{w \neq v} d_{\text{SP}}(v,w)}.
\end{equation}
The smaller the average distance between $v$ and the other nodes, the higher is the closeness of $v$. To better understand the meaning of closeness, let us consider the two graphs in Figure~\ref{fig:comparison}. Since closeness takes only shortest-path distances into account, the closeness of node $x_1$ in the graph on the left and the score of node $x_2$ in the graph on the right will be exactly the same. 
However, there is only one path connecting $x_1$ to each of the other nodes. This means that if just a single edge is removed from the graph, $x_1$ will become disconnected from part of the other nodes. For example, let us assume the edges represent streets and $x_1$ is the location of an ambulance. If a congestion occurs, the ambulance in $x_1$ will not be able to reach part of the nodes (or needs to take a long detour).
On the other hand, if the ambulance was in $x_2$, a congestion would limit only partially (or not at all) the ability of the ambulance to reach the nodes of the network.
\begin{figure}[tb]
  \begin{center}
    \includegraphics[width=0.4\textwidth]{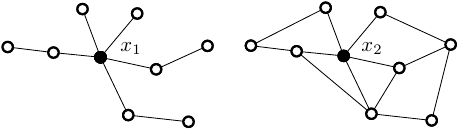}
  \end{center}
  \caption{Shortest-path closeness centrality cannot distinguish between node $x_1$ and node $x_2$.}
  \label{fig:comparison}
\end{figure}

The example above illustrates that traditional closeness is unable to model scenarios where the distance between two nodes does not only depend on the length of the shortest path between them, but also on the number of shortest or relatively short paths between the nodes. For this reason, a variant of closeness named \textit{current-flow closeness centrality} has been introduced~\cite{DBLP:conf/stacs/BrandesF05}. If we see the graph as an electrical network, the \textit{effective resistance} $d_{\text{ER}}(u,v) := p_{uv}(u) - p_{uv}(v)$ can be interpreted as an alternative distance measure between nodes $u$ and $v$. Indeed, if we multiply $d_{\text{ER}}(u,v)$ by the volume of $G$ (i.e. the sum of the weights of the edges in $G$), we get the \textit{commute time} between $u$ and $v$. The commute time between nodes $u$ and $v$ is defined as $H(u,v) + H(v,u)$, where the hitting time $H(x,y)$ is the expected time step in which a random walk in the graph starting in $x$ reaches $y$ for the first time. 
Thus, the commute time can be seen as the expected time a random walk needs for going from $u$ to $v$ and back again.
Since the commute time is based on random walks, it depends on \emph{all} the paths between two nodes. Thus, one can use the effective resistance to define a modified centrality measure~\cite{DBLP:conf/stacs/BrandesF05}:
\begin{equation}
\label{eq:closeness2}
c_{\text{ER}}(v) := \frac{n-1}{\displaystyle{\sum_{w \neq v} d_{\text{ER}}(v,w)}} = \frac{n-1}{\displaystyle{\sum_{w \neq v} p_{vw}(v) - p_{vw}(w)}},
\end{equation}
where the denominator is also called \emph{electrical farness}.
By convention, we define $d_{\text{ER}}(v,v) := 0\ \forall v \in V$.
To compute $c_{\text{ER}}(v)$ of a node $v$, we can solve $n-1$ linear systems. Alternatively, we could invert the Laplacian matrix $L$ of $G$ (after omitting the row and column corresponding to a node, in order to get a regular matrix $\tilde{L}$), using the property that $p_{vw}(v) - p_{vw}(w) = \tilde{L}^{-1}_{vv} - 2 \tilde{L}^{-1}_{vw} + \tilde{L}^{-1}_{ww}$~\cite{DBLP:conf/stacs/BrandesF05}.

\section{Approximating current-flow closeness}
\label{sec:approx}
As outlined in~\cite{mavroforakis2015spanning}, the fastest Laplacian linear solvers with a theoretical time complexity guarantee run in
$\tilde{O}(m \log{n} \log{(1/\tau)})$ time. Multigrid methods such as CMG and LAMG are much faster
in practice and have an \emph{empirical} running time of $O(m \log{(1/\tau)})$.

Computing current-flow closeness for only one node to the desired tolerance would already require
the solution of $n-1$ linear systems, yielding $O(n^2 \log{(1/\tau)})$ time in practice assuming a sparse graph.
This is infeasible for large networks with millions of nodes and edges. For this reason, we propose two approximation techniques for computing current-flow closeness for a subset of the nodes in large graphs, and we compare them in our experimental evaluation. The first one is based on a simple sampling approach, which recalls the one used for classical closeness~\cite{DBLP:journals/jgaa/EppsteinW04}. The second one uses the Johnson-Lindenstrauss transform (JLT), which allows to project the system into a lower-dimensional space by using $\bigO{\log n}$ random vectors. 


\subsection{Sampling-based approximation.}
\label{sec:sampling}
The idea is to sample uniformly at random a set $S \subseteq V$ of nodes $S = \{s_1, ..., s_k\}$, which we call \textit{pivots}. To approximate the current-flow closeness of a node $v$, we compute the effective resistance $d_{\text{ER}}(s, v)$ between all nodes $s \in S$ and $v$. Then, the closeness of $v$ can be approximated as 
\[ 
\tilde{c}_{\text{ER}}(v) := \frac{k}{n} \cdot \frac{n-1}{\tilde{f}_{\text{ER}}(v, S)},
\] 
where $\tilde{f}_{\text{ER}}(v, S) := \sum_{i = 1}^{k} d_{\text{ER}}(v, s_i)$ is the electrical farness of $v$
with respect to $S$.
\begin{proposition}
\label{prop:unbiased}
$\tilde{f}_{\text{ER}}(v, S)$ is un unbiased estimator for the electrical farness ${f}_{\text{ER}}(v)$ of $v$
(i.e.\ $E[\tilde{f}_{\text{ER}}(v, S)] = {f}_{\text{ER}}(v)$).
\end{proposition}
\begin{proof}
We show that $Y \coloneqq \frac{n}{k} \sum_{s_i\in S} d_{\text{ER}}(v, s_i)$ is an unbiased estimator for $s(v) = \sum_{w\in V} d_{\text{ER}}(v,w)$, then the claim follows directly from the properties of expected value. In the following, we denote the set of $k$-combinations of $V$ with $V_{k}$.
\[
\begin{split}
E(Y) & = \sum_{S = \{s_1, ..., s_k\} \in V_{k}} \frac{1}{\binom{n}{k}} \frac{n}{k} \sum_{s_i \in S} d_{\text{ER}}(v,s_i) = \\
	& = \frac{1}{\binom{n}{k}} \frac{n}{k} \sum_{w \in V} \binom{n-1}{k-1} d_{\text{ER}}(v,w) \\
	& = \sum_{w\in V} d_{\text{ER}}(v,w).\qquad \qquad \qquad \qquad \qquad \hfill \square
\end{split}
\]
\end{proof}
With $k$ pivots, the empirical complexity of our approach is $\bigO{k m \log (1/\tau)}$ with a multigrid solver.
Our experiments in Section~\ref{sec:approx_comp} show that a very small $k$ (e.g., $k=10$) is already enough to 
get a very good approximation.

\subsection{Projection-based approximation.} 
\label{sec:jlt}
Spielman and Srivastava~\cite{DBLP:journals/siamcomp/SpielmanS11} show how to compute an approximation of effective resistance based on the JLT. Let $B$ be the $m \times n$ incidence matrix where each row corresponds to an edge of $G$ and each node corresponds to a node such that, for edge $e = \{u,v\}$, $B(e,u) = +1$, $B(e,v) = -1$ and $B(e, w) = 0\  \forall w \neq u,v$ (since $G$ is undirected, the direction of edge $e$ can be chosen arbitrarily). Then, they show that the effective resistance between node $u$ and node $v$ can be re-written as $d_{\text{ER}}(u,v) = ||W^{1/2} B L^{\dagger} (e_u - e_v)||_2^2$, where $W$ is the diagonal $m \times m$ matrix such that $W(e,e) = \omega(e)$, $L^{\dagger}$ is the Moore-Penrose pseudoinverse~\cite{DBLP:books/daglib/0086372} of $L$ and $e_u$ is the $n \times 1$ vector such that $e(u) = 1$ and equal to 0 everywhere else. The effective resistances can therefore be seen as pairwise distances between vectors in $\{W^{1/2} B L^{\dagger} e_u \}_{u\in V}$, which allows to apply the JLT: If we project the vectors into a lower-dimensional space spanned by $k = \bigO {\log n}$ random vectors, the pairwise distances are approximately preserved. In other words, we can consider the pairwise distances between vectors in $\{QW^{1/2} B L^{\dagger} e_u \}_{u\in V}$, where $Q$ is a random projection matrix of size $k \times m$ with elements in $\{0, +\frac{1}{\sqrt{k}}, -\frac{1}{\sqrt{k}}\}$. 

Since we do not want to compute $QW^{1/2} B L^{\dagger}$ directly (it would require to (pseudo)invert $L$), we approximate it by solving $k$ linear systems: for $i = 1,..., k$, the $i$-th row $z_i^T$ of $QW^{1/2} B L^{\dagger}$ can be computed by solving the system $L z_i = \{QW^{1/2} B\}_{\cdot, i}$, see Algorithm~\ref{algo:JLT} (which we reuse from~\cite{DBLP:journals/siamcomp/SpielmanS11}). Note that the mul\-ti\-pli\-cation in Line~\ref{line:multiplication} requires only $\bigO{m \log n}$ operations, since $B$ is sparse (with $2m$ non-zero entries) and $W$ is diagonal.
\begin{algorithm2e}
 \begin{small}
\LinesNumbered
\SetKwData{B}{$\tilde{c}_B$}\SetKwData{VD}{VD}
\SetKwFunction{getVertexDiameter}{getVertexDiameter}
\SetKwFunction{sampleUniformNodePair}{sampleUniformNodePair}
\SetKwFunction{computeExtendedSSSP}{computeExtendedSSSP}
\SetKwInOut{Input}{Input}\SetKwInOut{Output}{Output}
\Input{$G=(V,E)$}
\Output{Approx. $\tilde{d}_{\text{ER}}(u,v)\ \forall (u, v)\in V \times V$}
Construct random matrix $Q$\;
Compute $Y = QW^{1/2} B$\; \label{line:multiplication}
$Z \leftarrow$ empty $k \times n$ matrix\;
\For{$i = 1,..., k$}
{
	solve the system $L z_i = Y_{\cdot, i}$\;
	$Z_{i, \cdot} \leftarrow z_i^T$\;
}
\ForEach{$(u,v) \in V \times V$}
{
$\tilde{d}_{\text{ER}}(u,v) \leftarrow ||Z_{\cdot, u} - Z_{\cdot, v} ||^2_2$\;
}
\Return{$\tilde{d}_{\text{ER}}$}
\end{small}
\caption{Effective resistance approximation~\cite{DBLP:journals/siamcomp/SpielmanS11}}
\label{algo:JLT}
\end{algorithm2e}
When choosing $k$ in Algorithm~\ref{algo:JLT} equal to $\bigO{\log n/ \epsilon^2}$ for any $\epsilon > 0$, it was shown~\cite{DBLP:journals/siamcomp/SpielmanS11} that, with probability $\geq 1-1/n$,
$$(1-\epsilon){d}_{\text{ER}}(v,w) \leq \tilde{d}_{\text{ER}}(v,w) \leq (1+\epsilon){d}_{\text{ER}}(v,w)$$
for all $(v,w)\in V\times V$.

An approximation of current-flow closeness for node $v$ can therefore be computed as $\tilde{c}_{\text{ER}}(v) :=  (n-1)/\sum_{w \neq v} \tilde{d}_{\text{ER}}(v, w)$.
If $(1-\epsilon){d}_{\text{ER}}(v,w) \leq \tilde{d}_{\text{ER}}(v,w) \leq (1+\epsilon){d}_{\text{ER}}(v,w)$ for each $w \neq v$, then also $(1-\epsilon){c}_{\text{ER}}(v) \leq \tilde{c}_{\text{ER}}(v) \leq (1+\epsilon){c}_{\text{ER}}(v)$. This is provably true only with probability $(1-1/n)^{n-1}$. However, our experimental results show that the approximation works well in practice: on all tested instances, $\tilde{c}_{\text{ER}}$ is \textit{always} within a $(1+\epsilon)$-factor from ${c}_{\text{ER}}$ 
(see Section~\ref{sec:approx_comp}).

\section{Experimental Evaluation}
\label{sec:experiments}

In this section we evaluate the performance of the two approximation algorithms described in Section~\ref{sec:approx}. First, we want to give some more details on the implementation, the benchmarking setup and the graph instances we used, before elaborating on the results of our evaluation.

\subsection{Implementation.}
We implemented both approximation algorithms in NetworKit~\cite{Staudt2014}, an open-source tool for fast exploratory analysis of massive networks. As became apparent in Section~\ref{subsec:LaplaciansInNetworkAnalysis}, linear systems play a quite important role in network analysis and since both approximation approaches introduced in this paper rely on solving Laplacian systems, we provide the Laplacian solver LAMG by Livne and Brandt~\cite{livne2012lean} in NetworKit with our own new implementation. 
The original implementation by Livne and Brandt is written in Matlab with some performance-critical parts in C and therefore
difficult for us to integrate into large-scale network analysis workflows. 

In informal experiments, our C++ implementation of LAMG outperforms their Matlab/C implementation regarding the solve times by a factor of 1.5 on average. In comparison the CMG solver by Koutis et al.~\cite{koutis2011combinatorial} is on average 11\% faster than our LAMG implementation on our large test instances. When solving linear systems, in all our experiments we set the relative residual error $\tau$ to $10^{-5}$.

\subsection{Benchmarking Setup.}
All experiments were done on a machine equipped with 256 GB RAM and a 2.7 GHz Intel Xeon CPU E5-2680 having 2 sockets with 8 cores each
and hyperthreading enabled.
The machine runs 64 bit SUSE Linux and we compiled our code with g++-4.8.1 and OpenMP~3.1.

\subsection{Instances.} \label{sub:instances}
Tables~\ref{table:small_instances}~and~\ref{table:instances} in the appendix show the set of instances we use for our experiments. While Table~\ref{table:small_instances} includes rather small complex networks with up to about 150\,000 edges, Table~\ref{table:instances} includes larger networks with up to 56 million edges. If a network has more than one connected component, we used the largest connected component (LCC). We ignore self-loops and the direction of edges in case a graph is directed. All the graphs are unweighted.
%
\subsection{Approximation algorithms.}
\label{sec:approx_comp}
In this section we compare the two approximation algorithms described in Section~\ref{sec:sampling} and Section~\ref{sec:jlt}, respectively. We refer to the first one as \textsc{Sampling} and to the second one as \textsc{Projection}. \textsc{Sampling} depends on the number $|S|$ of samples, whereas \textsc{Projection} depends on the dimension $k$ of the $k \times n$ random projection matrix. For simplicity, we call \textit{exact} the approach computing $c_{\text{ER}}$ as in Eq.~(\ref{eq:closeness2}), solving $n-1$ linear systems to the desired tolerance $\tau$.

For our experiments, we select 100 nodes for each of the networks shown in Table~\ref{table:small_instances} in the appendix. For each of these nodes, we compute current-flow ``exactly'' (to the desired tolerance $\tau$) and the two approximations with different parameters. In particular, we set the number $|S|$ of samples of \textsc{Sampling} to 10, 20, 50, 100, 200, 500, and 1000. When running \textsc{Projection}, we fix $k$ to $\lceil \log n / \epsilon ^2\rceil$ and set $\epsilon$ equal to 0.5, 0.2, 0.1, and 0.05. To measure the accuracy of the algorithms, we use the well-known Spearmann rank correlation coefficient, which measures how close the ranking of nodes determined by the approximation algorithm is close to that of the exact algorithm. We recall that the closer the Spearmann coefficient is to 1, the more correlated are the two rankings, with 0 meaning no correlation and 1 meaning the two ranks are identical. 

Figure~\ref{fig:approximations} reports the accuracy (Spearmann coefficient) and the running times in seconds for each approximation algorithm and for each parameter. We do not report explicitly to which parameter each point in the plot corresponds to, but this can be easily deduced from the running times: a smaller sample size corresponds to a smaller running time for \textsc{Sampling} and a larger $\epsilon$ corresponds to a smaller running time for \textsc{Projection}. Figure~\ref{fig:approximations} reports, for each approximation algorithm and for each parameter, the average over all networks of Table~\ref{table:small_instances} of time and Spearmann coefficient. 

The results are quite self-explanatory: the \textsc{Sampling} approach clearly outperforms \textsc{Projection} and its accuracy is extremely high already with only 10 samples. We also compute for each algorithm and parameter the number of rank inversions, i.e. the number of node pairs $\{u,v\}$ for which the approximated closeness of $u$ is smaller than the approximated closeness of $v$, but the exact closeness of $u$ is larger than or equal to the exact one of $v$ (or vice versa). With ten pivots, the average number of rank inversions of \textsc{Sampling} is 12.5; it is always below 10 for higher number of samples. This means that, out of $\binom{100}{2} = 4950$ pairs, less than 10 are inverted, corresponding to $0.2\%$.
\begin{figure}[htb]
  \begin{center}
    \includegraphics[width=0.5\textwidth]{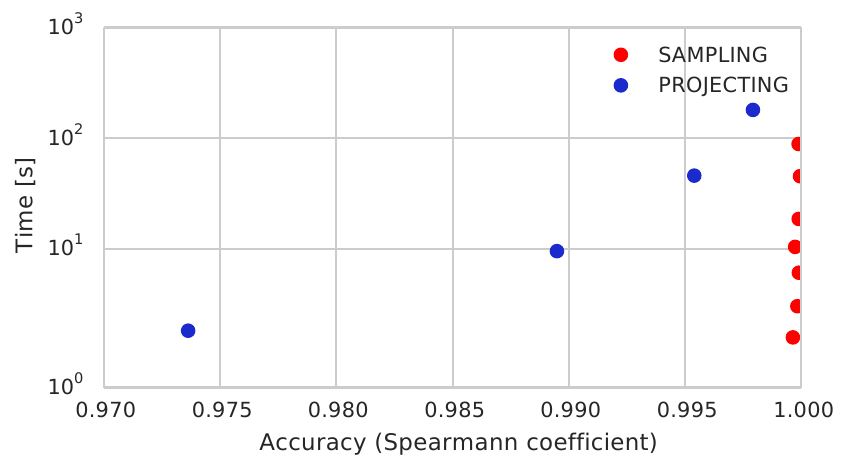}
  \end{center}
  \vspace{-3ex}
  \caption{Time vs. Spearmann coefficient for the two approximation algorithms, using different parameters. The points represent the average among the networks of Table~\ref{table:small_instances} in the Appendix.}
  \label{fig:approximations}
\end{figure}

In addition to accuracy in terms of ranks, we also evaluate the maximum relative error. We define the relative error for a node $v$ as $e(v) = \max\{r(v), 1/r(v)\}$, where $r(v)$ is the ratio between the exact current-flow closeness of $v$ and its approximation. The maximum relative error is then defined as $\max_{v \in V}{e(v)}$. Figure~\ref{fig:rel_errors} reports the results. It is interesting to notice that, with respect to this measure, the two algorithms behave quite similarly. 
Also, notice that the maximum relative error for \textsc{Projection} is always smaller than $\epsilon$ (we recall the values of $\epsilon$ used are 0.05, 0.1, 0.2 and 0.5), although we can only prove that this is true with probability at least $(1-1/n)^{n-1}$.

To summarize, our results show that both algorithms lead to very good accuracy in terms of maximum relative error, while the sampling approach better preserves the ranking of nodes, even when the number of samples is very small. For this reason, in our experiments on large graphs, we make use of the sampling approach.
\begin{figure}[htb]
  \begin{center}
    \includegraphics[width=0.5\textwidth]{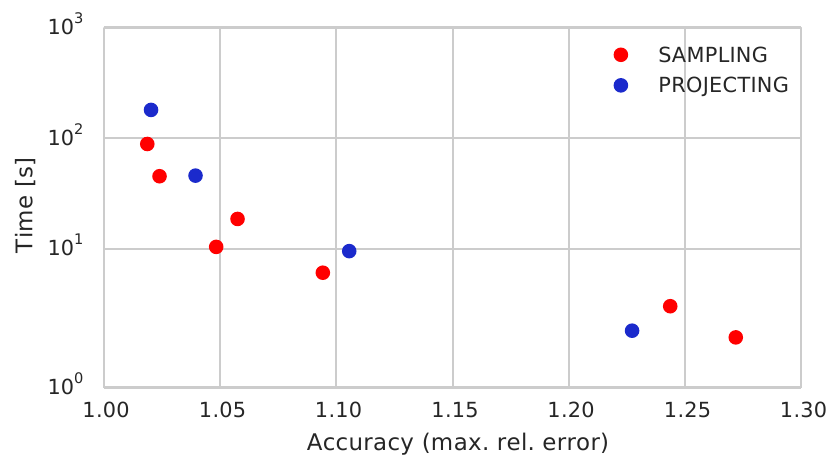}
  \end{center}
  \vspace{-3ex}
  \caption{Time vs. maximum relative error for the two approximation algorithms, using different parameters. The points represent the average among the networks of Table~\ref{table:small_instances} in the Appendix.}
  \label{fig:rel_errors}
\end{figure}
On average (over the instances of Table~\ref{table:small_instances}), computing $c_\text{ER}$ on 100 nodes takes more than 20 minutes, whereas using \textsc{Sampling} with 20 pivots takes only 2.87 seconds. Table~\ref{table:add_results} in the Appedix shows the detailed running times.

\subsection{Comparison with shortest-path closeness.}
As explained in Section~\ref{sec:cf_closeness}, our intuition is that current-flow closeness should represent the efficiency of a node reaching the other nodes of the network better than shortest-path closeness. To verify this assumption, we first compare the two measures in terms of their capability to discriminate between different nodes. In this experiments, we use the networks of Table~\ref{table:small_instances} and compute (exactly) current-flow and shortest-path closeness on 100 randomly chosen nodes.
Figure~\ref{fig:rel_std_dev} shows the relative standard deviation for shortest-path and current-flow closeness. The relative standard deviation is defined as the standard deviation divided by the average. It is always significantly higher for current-flow closeness than it is for shortest-path closeness, meaning that there is much more variation in the scores computed by the former.
\begin{figure}[htb]
  \begin{center}
    \includegraphics[width=0.5\textwidth]{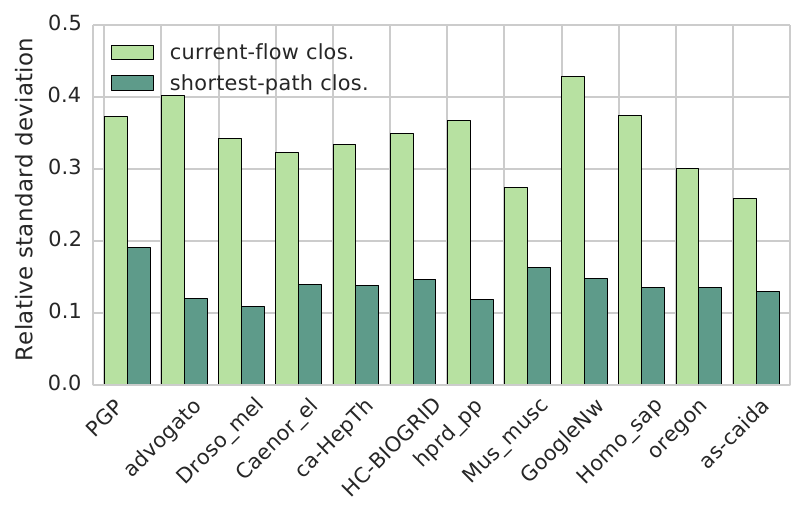}
  \end{center}
   \vspace{-5ex}
  \caption{Relative standard deviation for shortest-path and current-flow closeness.}
  \label{fig:rel_std_dev}
\end{figure}

Also, similarly to what has been done in~\cite{mavroforakis2015spanning} for spanning edge centrality and edge betweenness centrality, we measure the resilience to noise, in this case for current-flow closeness and shortest-path closeness. The idea is to add edges to the graph and see how well the initial rankings are preserved. Our intuition is that, if we add some edge that creates a shortcut between a node $v$ and some other nodes, the shortest-path closeness of $v$ will be more affected than its current-flow closeness, since the former takes only shortest paths into account. This is confirmed by our experiments, summarized in Figure~\ref{fig:resiliance}. For each network in Table~\ref{table:small_instances}, we insert a percentage of the total number of edges varying from $1\%$ to $10\%$. To have a high number of shortcuts involving the sampled nodes, we always add edges between one of the sampled nodes and other nodes of the graph. Figure~\ref{fig:resiliance} shows, for each percentage of inserted edges, the average among all tested networks of the Spearmann correlation coefficient between the initial ranking and the ranking after the insertions. Figure~\ref{fig:resiliance} shows that current-flow closeness is more resilient to edge insertions and the difference between the resilience of the two measures increases the more the graph changes. 
\begin{figure}[htb]
  \begin{center}
    \includegraphics[width=0.5\textwidth]{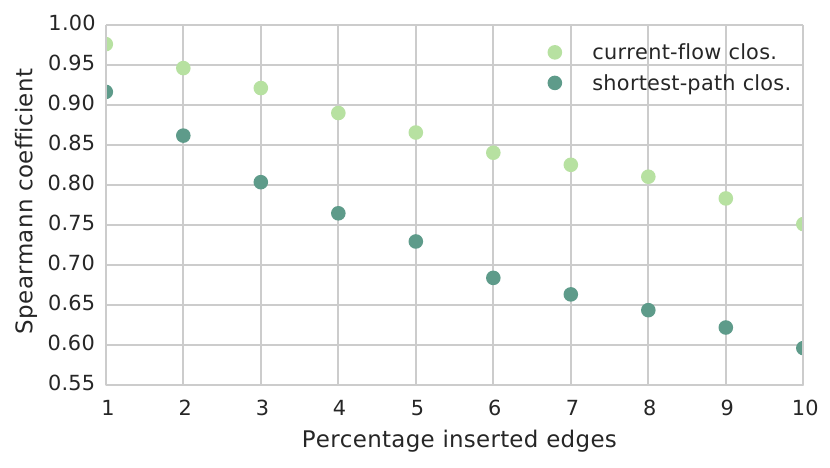}
  \end{center}
  \vspace{-3ex}
  \caption{Resilience to noise for different percentages of inserted edges. The points represent the average among the networks of Table~\ref{table:small_instances} in the Appendix.}
  \label{fig:resiliance}
\end{figure}
%
\subsection{Correlation with degree.}
\label{ref:corr_degrees}
In certain random geometric graph models, such as $\epsilon$-graphs, kNN graphs, and Gaussian similarity graphs, it was recently shown~\cite{DBLP:journals/jmlr/LuxburgRH14} that the effective resistance between two nodes $u$ and $v$ converges to $1/\deg{(u)}+1/\deg{(v)}$ when the number of nodes goes to infinity. This result also has implications on current-flow closeness: when the number of nodes goes to infinity in such graphs, $c_{\text{ER}}(v)$ goes to $c_{\text{A}}(v): = (n-1)/\sum_{w\neq v}(1/\deg{(v)} + 1/\deg{(w)} )$. 
However, the structure of real-world networks (e.g. complex or street networks) is significantly different from random graphs and it is not clear how close current-flow closeness is to this asymptotic value in reality. In this section we therefore study the correlation between current-flow closeness, as well as shortest-path closeness and betweenness, with $c_{\text{A}}$. In our experiments, we consider large networks with up to 56 millions nodes and edges and we use the sampling approach to approximate current-flow closeness (with 20 pivots). Table~\ref{table:times} shows the running times of our approximation when computing the closeness of a single node.
We approximate betweenness using the approach presented in~\cite{DBLP:conf/alenex/GeisbergerSS08}, which is already implemented in NetworKit. Since the results are very different for street and complex networks, we study them separately.

Figure~\ref{fig:corr_degrees} shows the Spearmann coefficient computed between the three centrality measures and $c_{\text{A}}$ on the complex networks of Table~\ref{table:instances} in the Appendix. The results show that there is in fact a strong positive correlation. This is weaker for shortest-path closeness, with an average Spearmann coefficient of 0.63 and stronger for betweenness and current-flow closeness, with an average of 0.81 and 0.89, respectively. We obtain similar results on the smaller instances of Table~\ref{table:small_instances}, where the averages are 0.63, 0.85 and 0.89 for shortest-path closeness, betweenness and current-flow closeness, respectively.
\begin{figure}[htb]
  \begin{center}
    \includegraphics[width=0.5\textwidth]{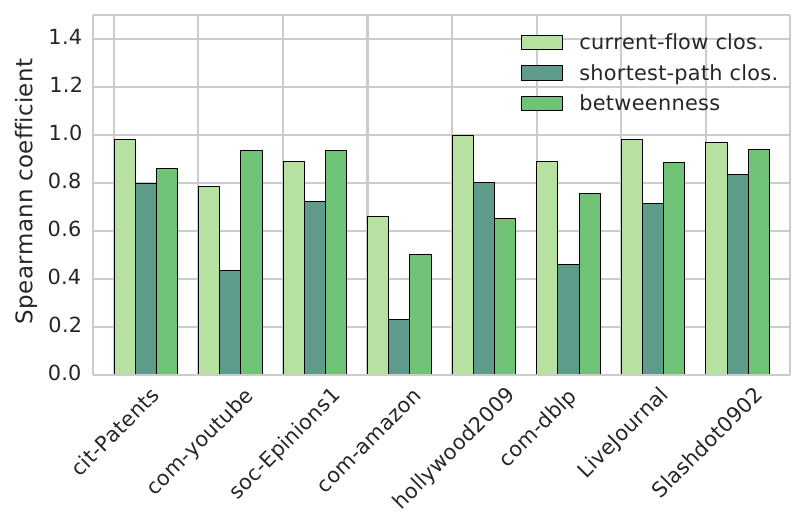}
  \end{center}
  \vspace{-5ex}
  \caption{Correlation with $c_{\text{A}}$ for complex networks.}
  \label{fig:corr_degrees}
\end{figure}
The results are very different for street networks (Figure~\ref{fig:corr_degrees_streets}). Here the correlation with the degrees is generally very low and sometimes even negative, with an average of -0.02 for closeness, 0.35 for betweenness and 0.07 for current-flow closeness.

While $c_{\text{A}}$ and $c_{\text{ER}}$ are unrelated on street networks, our results show that there is actually a strong correlation between them in complex networks. 
This behavior is likely due to the different type of degree distributions in the two network classes. While complex networks
usually feature a skewed degree distribution with many small, but also some high-degree nodes, the degrees in street networks
are closely concentrated.
\begin{table}
\caption{Running time of \textsc{Sampling} with 20 pivots when computing $c_{\text{ER}}$ of a single node.}
\label{table:times}
\begin{tabularx}{8cm}{l r}	
\hline
Graph & Time approximation [s] \\
\hline
cit-Patents & 125.99 \\
com-youtube & 5.06 \\
soc-Epinions1 & 0.28 \\
com-amazon & 2.56 \\
hollywood2009 & 107.52 \\
com-dblp & 1.90 \\
LiveJournal & 287.27 \\
Slashdot0902 & 0.41 \\
roadNet-TX & 6.28 \\
luxembourg & 0.13 \\
belgium & 2.39 \\
netherlands & 5.33 \\
italy & 10.91 \\
great-britain & 12.72 \\
europe & 103.34 \\
\end{tabularx}
 \vspace{-4ex}
\end{table}
Consequently, in some applications where a very good accuracy is not needed, $c_{\text{A}}$ might be used as an approximation of $c_{\text{ER}}$ in complex networks. The same thing can be said for betweenness, which is only slightly less correlated with $c_{\text{A}}$ than $c_{\text{ER}}$. This is very convenient, since $c_{\text{A}}$ can be computed in $\bigO{m}$ time. However, our results in Section~\ref{sec:approx_comp} show that our sampling-based approach can compute an extremely accurate approximation in time $\bigO{k m \log (1/\epsilon)}$ even when the number $k$ of samples is very small. For this reason, we believe the sampling approach is probably the best option for most applications.
\begin{figure}[tb]
  \begin{center}
    \includegraphics[width=0.5\textwidth]{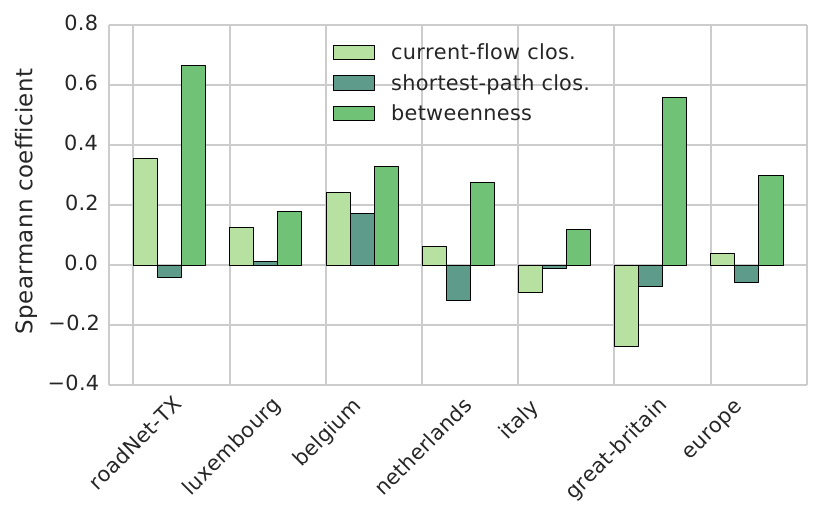}
  \end{center}
  \vspace{-5ex}
  \caption{Correlation with $c_{\text{A}}$ for street networks.}
  \label{fig:corr_degrees_streets}
\end{figure}

\section{Conclusions}
%
Although many important graph properties can be formulated in terms of Laplacian linear systems, popular network analysis frameworks lack an implementation of state-of-the-art solvers. We bridge this gap by providing a Lean Algebraic Multigrid (LAMG) implementation in NetworKit, our toolkit for large-scale network analysis.

Based on our new LAMG implementation and our algorithms \textsc{Sampling} and \textsc{Projection}, we have computed current-flow closeness centrality and provided the first published results on its behavior on large real-world networks. Our algorithms lead to very accurate results and, thanks to them, we are now able to compute an estimation of current-flow closeness of a reasonably small subset of nodes on networks with tens of millions of nodes and edges within a few seconds or minutes.
In our experiments current-flow closeness alleviates two known problems of shortest-path closeness and can thus be
seen as a viable alternative in many scenarios.
We have also shown empirically that there is a strong correlation between degrees and both current-flow closeness and betweenness centrality in complex networks, whereas the degree and current-flow closeness are basically unrelated in street networks.

Our approach based on \textsc{Sampling} or \textsc{Projection} is very fast in scenarios where we only need to compute the current-flow closeness for a subset of nodes. 
Yet, especially \textsc{Sampling} is probably too expensive if closeness has to be computed for all nodes. 
For this case, it is an interesting aspect of future work to improve on JLT and pseudoinversion.

\subsubsection*{Acknowledgments}
This work was partially supported by German Research Foundation (DFG) grant ME3619/3-1
within Priority Programme 1736 \textit{Algorithms for Big Data}.

\begin{small} 
\bibliographystyle{abbrv}
\bibliography{references.bib}
\end{small}
\clearpage
\appendix
\section{Instances used for the experiments}
\begin{minipage}{\textwidth}
\vspace{2ex}
	\captionof{table}{Properties of smaller benchmark instances used in this paper.} \label{table:small_instances}
	\begin{tabularx}{\textwidth}{l r r X r}	
	\hline
	Graph & \#Nodes in LCC & \#Edges in LCC & Description & Ref. \\
	\hline
	PGP & 10680 & 24316 & PGP trust network & ~\cite{BaderMSW12dimacs} \\
	advogato& 5272 & 42816 & Advocato trust network &~\cite{Lasagne} \\
	Drosophila\_melanogaster & 10424 & 40660 & Interactome & ~\cite{Lasagne} \\
	Caenorhabditis\_elegans & 4428 & 9659 & Metabolic network & ~\cite{Lasagne} \\
	CA-HepTh & 8638 & 24806 & Collaboration Network & \cite{SNAP} \\
	HC-BIOGRID & 4039 & 10321 & Genetic interaction &~\cite{Lasagne} \\
	hprd\_pp & 9219 & 36900 & Human proteine interaction & ~\cite{Lasagne} \\
	Mus\_musculus & 3745 & 5170 & Interactome & ~\cite{Lasagne} \\
	GoogleNw & 15763 & 148585 & Hyperlinks between web pages &~\cite{Lasagne} \\
	Homo\_sapiens & 13478 & 61006 & Metabolic network & ~\cite{Lasagne} \\
	oregon2\_010526 & 11461 & 32730 & AS peering network & \cite{SNAP} \\
	as-caida20071105 & 26475 & 53381 & CAIDA AS relationships & \cite{SNAP}   \\
	\end{tabularx}	
\vspace{2ex}
	\captionof{table}{Properties of larger benchmark instances used in this paper.} \label{table:instances}
	\begin{tabularx}{\textwidth}{l r r X r}	
	\hline
	Graph & \#Nodes in LCC & \#Edges in LCC & Description & Ref. \\
	\hline
	cit-Patents & 3764117 & 16511740 & Citation Network & \cite{SNAP} \\
	com-Amazon & 334863 & 925872 & Amazon Product Network & \cite{SNAP} \\
	com-DBLP & 317080 & 1049866 & Collaboration Network & \cite{SNAP} \\
	com-Youtube & 1134890 & 2987624 & Youtube Social Network & \cite{SNAP} \\
	hollywood-2009 & 1069126 & 56306653 & Collaboration Network & \cite{Davis2011} \\
	com-LiveJournal & 3997962 & 34681189 & LiveJournal Social Network & \cite{SNAP} \\
	Slashdot0902 & 82168 & 504230 & Slashdot Zoo Social Network & \cite{SNAP} \\
	soc-Epinions1 & 75877 & 405739 & Epinions Social Network & \cite{SNAP} \\
	roadNet-TX & 1351137 & 1879201 & Road Network of Texas & \cite{SNAP} \\
	luxembourg.osm & 114599 &119666 & Road Network of Luxembourg & ~\cite{BaderMSW12dimacs} \\
	belgium.osm & 1441295 & 1549970 & Road Network of Belgium & ~\cite{BaderMSW12dimacs} \\
	netherlands.osm & 2216688 & 2441238 & Road Network of the Netherlands & ~\cite{BaderMSW12dimacs} \\
	italy.osm & 6686493 & 7013978 & Road Network of Italy & ~\cite{BaderMSW12dimacs} \\
	great\_britain.osm & 7733822 & 8156517 & Road Network of Great Britain & ~\cite{BaderMSW12dimacs} \\	
	europe.osm & 50912018 & 54054660 & Road Network of Europe & ~\cite{BaderMSW12dimacs} \\	
	\end{tabularx}	
\end{minipage}
\clearpage
\section{Additional experimental results}
\begin{minipage}{\textwidth}
\vspace{2ex}
\captionof{table}{Comparison between exact (= within desired tolerance $\tau$) and \textsc{Sampling} approach, with 20 pivots. The first two columns report the running times of the two approaches, when computing current-flow closeness on 100 nodes. The third column reports the Spearmann rank correlation coefficient between the approaches and the fourth the percentage of rank inversions.}\label{table:add_results}
	\begin{tabularx}{\textwidth}{l r r r r}	
	\hline
	Graph & Time exact [s] & Time \textsc{Sampling} 20 [s] & Spearmann coeff. & Rank Inver. \\
	\hline
	PGP & 558.97 & 1.68 & 0.99990 & 0.14\% \\
	advogato & 383.42 & 2.39 & 0.99986 & 0.16\% \\
	Drosophila\_melanogaster & 1077.78 & 3.50 & 0.99986 & 0.12\% \\
	Caenorhabditis\_elegans & 68.50 & 0.64 & 0.99975 & 0.28\% \\
	CA-HepTh & 800.65 & 2.87 & 0.99989 & 0.14\% \\
	HC-BIOGRID & 186.47 & 1.92 & 0.99975 & 0.28\% \\
	hprd\_pp & 988.58 & 4.01 & 0.99990 & 0.10\% \\
	Mus\_musculus & 33.66 & 0.33 & 0.99958 & 0.44\% \\
	GoogleNw & 4612.19 & 8.16 & 0.99987 & 0.14\% \\
	Homo\_sapiens & 1913.06 & 4.90 & 0.99999 & 0.02\% \\
	oregon2\_010526 & 640.97 & 1.49 & 0.99988 & 0.14\% \\
	as-caida20071105 & 3354.62 & 2.62 & 0.99990 & 0.12\% \\
	\end{tabularx}	
\end{minipage}

\end{document}